\documentclass[11pt,a4paper]{article}

\usepackage[T1]{fontenc}
\usepackage[utf8]{inputenc}
\usepackage[english]{babel}

\usepackage{geometry}
\usepackage{graphicx}
\usepackage{amssymb,amsmath,amsthm}
\usepackage{stmaryrd}
\usepackage{delarray}
\usepackage{thmtools}
\usepackage{thm-restate}
\usepackage{enumitem}
\setlist{nosep}
\usepackage{tikz}
\usepackage{hyperref}
\usepackage[capitalise]{cleveref}
\graphicspath{{figures/}}
\usepackage{authblk}
\usepackage{wrapfig}
\usepackage{multicol}
\usepackage{float}
\usepackage{appendix}
\usepackage{pdfpages}
\usepackage{todonotes}
\usepackage{qtree}

\newtheorem{theorem}{\textbf{Theorem}}
\newtheorem*{theorem*}{\textbf{Theorem}}
\newtheorem{lemma}[theorem]{\textbf{Lemma}}
\newtheorem{proposition}[theorem]{\textbf{Proposition}}
\newtheorem*{proposition*}{\textbf{Proposition}}

\newtheorem{corollary}[theorem]{\textbf{Corollary}}
\newtheorem{remark}{\textbf{Remark}}

\newtheorem{definition}{\textbf{Definition}}

\newtheorem{remark*}{Remark}

\renewcommand{\int}[1]{[#1]}
\newcommand{\intz}[1]{\llbracket#1\rrbracket}
\newcommand{\N}{\mathbb{N}}

\newcommand{\Cc}{\mathcal{C}}
\newcommand{\Dc}{\mathcal{D}}
\renewcommand{\O}{\mathcal{O}}
\newcommand{\Pc}{\mathcal{P}}

\newcommand{\Poly}{{\mathsf{P}}}
\newcommand{\NP}{\ensuremath{\mathsf{NP}}}
\newcommand{\coNP}{\ensuremath{\mathsf{coNP}}}

\newcommand{\spec}{\mathsf{Spec}}
\newcommand{\col}{\mathsf{Col}}

\newcommand{\dyna}[1]{\mathcal{G}_{#1}}
\newcommand{\C}{\mathcal{C}}
\newcommand{\glue}{\triangleright}

\newcommand{\const}{\mathsf{constant}}
\newcommand{\recol}{\mathsf{recolor}}
\newcommand{\join}{\mathsf{join}}
\newcommand{\leftdir}{\mathsf{left}}
\newcommand{\rightdir}{\mathsf{right}}

\newcommand{\decisionpbpromisew}[5]{
  \fbox{\parbox{{#5}\textwidth}{{\bf {#1}}\\{\it Input:} {#2}\\{\it Promise:} {#3}\\{\it Question:} {#4}}}
}

\title{Complexity lower bounds for succinct binary structures of bounded clique-width with restrictions}
\author[a]{Colin Geniet}
\author[b]{Ali\'enor Goubault-Larrecq}
\author[b,c]{K{\'e}vin Perrot}
\affil[a]{Discrete Mathematics Group, Institute for Basic Science, Daejeon, South Korea}
\affil[b]{Aix Marseille University, CNRS, LIS, Marseille, France}
\affil[c]{Université publique, France}
\date{}

\begin{document} 
	\maketitle
	
	\begin{abstract}
	  We present a Rice-like complexity lower bound for any MSO-definable problem
          on binary structures succinctly encoded by circuits.
          This work extends the framework recently developed as a counterpoint to Courcelle's theorem
          for graphs encoded by circuits, in two interplaying directions:
          (1) by allowing multiple binary relations, and (2) by restricting the interpretation of new symbols.
          Depending on the pair of an MSO problem~$\psi$ and an MSO restriction~$\chi$,
          the problem is proven to be NP-hard or coNP-hard or P-hard,
          as long as~$\psi$ is non-trivial on structures satisfying~$\chi$ with bounded \emph{clique-width}.
          Indeed, there are P-complete problems (for logspace reductions) in our extended context.
          Finally, we strengthen a previous result on the necessity to
          parameterize the notion of non-triviality, hence supporting the choice of clique-width.
	\end{abstract}


	\section{Introduction}
        \label{sintro}

        Courcelle's celebrated theorem states that any question on graphs expressible
        in monadic second order logic (MSO) on signature $\{=,\to\}$
        (quantifiers range over vertices and sets of vertices, and $\to$ stands for the graph's binary relation)
        is solvable in linear time on any class of graphs
        of bounded clique-width~\cite{CourcelleMakowkyRotics2000,OumSeymour2006}.
        On the other way around, a recent line of works~\cite{ggpt21,gglgopt25} has established hardness results
        when the graph is not given by an adjacency matrix,
        but is succinctly encoded in a Boolean circuit which,
        given two input vertices $x,y$, outputs one bit telling whether $x\to y$
        holds, a representation introduced in~\cite{Galperin1983}.
        In the spirit of Rice's theorem in computability theory,
        such results provide general algorithmic complexity lower bounds for any non-trivial MSO sentence.
	Precisely, one requires the sentence to have an infinity of models and of counter-models, both of bounded clique-width; this is called being \emph{cw-non-trivial}.

	\begin{theorem}[\cite{gglgopt25}]\label{thm:hardness-cw-graphs}
          If $\psi$ is a cw-non-trivial MSO sentence,
          then testing $\psi$ on graphs represented succinctly is either $\NP$-hard or $\coNP$-hard.
        \end{theorem}
	This succinct representation of graphs by circuits is motivated
        by natural models of computation called automata networks.
        An automata network (AN) of size $n$ is a
        finite discrete dynamical systems composed of interacting entities.
	Each component $i\in\{1,\dots,n\}$ has a finite set of states $A_i$,
        and the configuration space is  $X=\prod_{i=1}^n A_i$.
        The behavior of each component is provided as a local circuit describing its behavior,
        and the tuple of such circuits gives a succinct representation of the dynamics
        (a graph on vertex set $X$).
        When $A_i=\{0,\dots,q-1\}$ for all $i$ we call it a $q$-uniform AN
        (and a Boolean AN when $q=2$), which enforces the dynamics to have $q^n$ configurations.
	The computational complexity of many individual problems has been settled,
	from the existence of fixed points~\cite{Alon1985,Floreen1989} and limit cycles~\cite{Bridoux2021},
	to reachability~\cite{Chatain2020} and limit dynamics~\cite{Pauleve2020}.
	Arithmetical constraints are solved in~\cite{glp25v2}, hence the Theorem above
        holds in the context of ANs and provides a universal lower bounds:
	all cw-non-trivial problems are $\NP$-hard or $\coNP$-hard.
        For deterministic systems (graphs of out-degree one, which are almost trees), the dichotomy is sharp
        because trivial questions are answered in constant time.

        In the present work, we focus on embracing larger contexts.
	We consider multiple binary relations encoded in one structure,
	with additional restrictions in order to enforce the meaning of the additional symbols
	(for example a partial or total order) or structural properties of the instances
	(for example that it is injective, or deterministic, or strongly connected).
	The restriction itself is also expressed in terms of an MSO sentence denoted~$\chi$, treated as a promise.
	We call {\bf $\psi$-under-$\chi$-dynamics} this decision problem,
	i.e.\ testing the formula~$\psi$ under the promise~$\chi$, given a graph succinctly encoded by a circuit.
	The non-triviality assumption is adapted naturally when adding the constraint~$\chi$:
	we ask~$\psi$ to be \emph{cw-non-trivial under the restriction~$\chi$},
	meaning that there should be infinitely many models and counter-models of~$\psi$
	which furthermore satisfy~$\chi$.

	When the restriction~$\chi$ is union-stable
	(i.e.\ if~$G_1,G_2$ satisfy~$\chi$, then so does their disjoint union $G_1 \sqcup G_2$)
	we generalize the proof technique developed in~\cite{gglgopt25}
	and obtain hardness at the first level of the polynomial hierarchy for any non-trivial formula~$\psi$.
	\begin{restatable}{theorem}{mainNP}\label{thm:main-NP-hard}
	  Let~$\psi$ and~$\chi$ be two MSO formulae, such that
	  $\psi$ is cw-non-trivial under the restriction~$\chi$, and $\chi$ is union-stable.
	  Then {\bf $\psi$-under-$\chi$-dynamics} is either $\NP$-hard or $\coNP$-hard.
	\end{restatable}
	The assumption that~$\chi$ is stable under disjoint unions is crucially used
	in the construction of \emph{MSO-saturating graphs}.

	Without this assumption however, the result fails:
	consider for instance as restriction~$\chi$ `the graph~$G$ is either a clique or edgeless',
	and as property~$\psi$ to test `$G$ is a clique'.
	Then testing~$\psi$ under the promise of~$\chi$ is equivalent to picking vertices~$x,y$ arbitrarily and testing if~$xy$ is an edge.
	With a succinct representation, this means evaluating the circuit representing~$G$ in~$(x,y)$,
	and is thus $\Poly$-complete (under logspace reductions).
	Using the logspace strengthening of the reductions developed in~\cite{glp25circuit},
	we show that the $\Poly$-hardness found in this example in fact always holds,
	except for the degenerate case where~$\psi$ under the promise of~$\chi$
	only depends on the size of the graph.

	Formally, we call $(\psi,\chi)$ \emph{cw-size-independent} if there are infinitely many sizes~$n$
	for which there is both a model and a counter-model of~$\psi$ of size~$n$, both satisfying~$\chi$.
	This in particular implies that~$\psi$ is cw-non-trivial under the restriction~$\chi$.
	\begin{restatable}{theorem}{mainP}\label{thm:main-P-hard}
	  Let~$\psi$ and~$\chi$ be two MSO formulae such that $(\psi,\chi)$ is cw-size-independent.
	  Then {\bf $\psi$-under-$\chi$-dynamics} is $\Poly$-hard under logspace reductions.
	\end{restatable}
	Note that without the cw-size-independence assumption,
	one can find~$\psi$ cw-non-trivial under~$\chi$ such that {\bf $\psi$-under-$\chi$-dynamics}
	can be solved in logspace by simply looking at the number of vertices of the input graph
	(which is given in the succinct model).
	For instance, choose~$\chi$ to require the graph~$G$ to be a cycle, and~$\psi$ to test whether~$G$ is bipartite,
	hence whether the number of vertices is odd or even.

	The previous hardness results all assume that the property~$\psi$ is non-trivial on bounded clique-width structures.
	It is natural to ask whether this can be relaxed.
	When one only simply asks~$\psi$ to have infinitely many models and counter-models,
	a negative answer is given by~\cite{gglgopt25}:
	when considering a non-trivial FO formula~$\psi$ whose spectrum (i.e.\ set of possible sizes of models) is extremely sparse,
	they show that a $\NP$- or $\coNP$-hardness result for~$\psi$ in the style of \cref{thm:hardness-cw-graphs} is very unlikely,
	as it would imply that there is a `large' set~$S$ of integers such that
	\textbf{SAT} can be solved in polynomial time for any instance with size in~$S$.
	The `large' set~$S$ here is formalized by the notion of \emph{robustness} of Fortnow and Santhanam~\cite{FS2017robust}.

	We prove that this result still holds with a much stronger non-triviality assumption on~$\psi$,
	namely that it is non-trivial for planar graphs with bounded degree.
	\begin{restatable}{theorem}{planar}\label{thm:intro-planar}
	  There is a first-order formula~$\psi$ which has infinitely many models and counter models among planar graphs with maximum degree~4, and such that
	  \textbf{$\psi$-dynamics} cannot be $\NP$- or $\coNP$-hard,
	  unless there is a polynomial time algorithm solving $\mathbf{SAT}$ on a robust set of sizes of instances.
	\end{restatable}
	Planar graphs with bounded degree are arguably the simplest graph classes beyond the realm of bounded clique-width.
	Indeed, they correspond to the simplest obstructions to clique-width --- namely planar grids.
	Furthermore, specifically from the point of view of FO logic,
	bounded degree graphs on the one hand and planar graphs on the other
	are two rather well-behaved graph classes,
	as illustrated by early results on FO model-checking for these graphs~\cite{seese1996linear,frick2001FOplanar}.
	Thus we believe \cref{thm:intro-planar} to be a very strong indication that
	\cref{thm:hardness-cw-graphs,thm:main-NP-hard} cannot be meaningfully generalized beyond the bounded clique-width setting.

        Given that the succinct version of FO problems can be complete for all levels of
        the polynomial hierarchy~\cite{ggpt21} and that MSO sentences go higher~\cite{Balcazar1992},
        it remains to understand more finely the level of complexity of some questions
        (the $\NP$ or $\coNP$ bound, although one can identify which of the two holds, is sometimes low).
        Transposing the results to restricted families of local functions or relations in ANs,
        and to other finite models of computations, would be meaningful.
        The aim is to prove that any question on the dynamics of finite discrete systems
        which is not trivial is complex.

	\paragraph*{Structure of the paper}
        Section~\ref{s:def} introduces the necessary definitions:
        succinct graph representations, clique-width, MSO types and compositionality.
	Then, we prove two different pumping lemmas in Section~\ref{s:pumpings},
	depending on the properties of the pair $(\psi,\chi)$.
	In Section~\ref{s:reductions}, we construct circuits to succinctly encode
	the graphs constructed by these pumping lemmas, giving complexity reductions which prove \cref{thm:main-NP-hard,thm:main-P-hard}.
        Section~\ref{s:applications} showcases the applicability of this general complexity lower bound,
        and provides examples to sustain the definitions at stake.
	Finally, in Section~\ref{s:bounded}, we prove \cref{thm:intro-planar}, strengthening the result of~\cite{gglgopt25}.

	\section{Definitions}
        \label{s:def}

	We denote by $\intz{n}=\{0,\dots,n-1\}$ the interval of integers.
        
        \subsection{Structures}

	A directed graph $G=(V,E)$ is a binary relation $E \subseteq V \times V$.
        In this work we consider more generally a set $R_2$ of binary relation symbols,
	and see them as	\emph{labels} for the arrows:
	an arc is a triple~$(a,u,v)$ where~$a \in R_2$ is the symbol of a relation.
	We thus consider \emph{$R_2$}-graphs $G=(V,E)$ with $E \subseteq R_2 \times V \times V$,
	which we sometimes simply call graphs.
	We denote by $|G|=|V|$ the number of vertices.

	A \emph{Boolean circuit} is a directed acyclic graph of in-degree and out-degree at most two.
	Vertices of in-degree zero are its \emph{inputs}, of in-degree one are $\neg$-gates,
	of in-degree two are either $\wedge$-gates or $\vee$-gates,
	and of out-degree zero are its \emph{outputs} (vertices are ordered and some label gives their kind).
	From an input binary string providing a value for each input,
	a circuit computes an output binary string following the usual semantics at each gate.

	A \emph{succinct representation} of an $R_2$-graph is a couple $(N,D)$
	when $N$ is the number of vertices and $D$ is a circuit with $2\cdot\lceil\log_2(N)\rceil$ inputs
	and $|R_2|$ outputs, computing the relations of $R_2$ which hold for a given pair of vertices $v,v'\in V$.
	We assume that the vertices are numbered from $0$ to $N-1$, and denote $\dyna{N,D}$ the encoded $R_2$-graph.
        The size of $D$ is considered to be polynomial in
        the size of adjacency matrices for each relation.
        
        \subsection{Clique-width}

        Let $C$ be a finite set of colors.
        First, we define operations meant to decompose $C$-colored $R_2$-graphs,
	which are $R_2$-graphs $G=(V,E,\col_G)$ with a coloring function $\col_G:V\to C$
	(each vertex has a unique color).
        
        \begin{itemize}
		\item $\const^\mathcal{R}_c(v)$: parameterized by a subset $\mathcal{R} \subseteq R_2$ and a color $c \in \col$,
		it creates from the vertex~$v$ the graph
        	$G=(\{v\}, \{(r, v, v) \mid r \in \mathcal{R}\}, \col_G)$ with
        	$\col_G(v) = c$.
		\item $\recol_f(H)$: for every \emph{recoloring function} $f : C \to C$
        	and $C$-colored graph $H$ it constructs the graph
        	$G = (V(H), E(H), \col_G)$ where $\col_G = f \circ \col_H$.
		\item $\join_M(H,H')$: parameterized by a subset $M \subseteq R_2 \times \{\leftdir,\rightdir\} \times C^2$,
		it constructs the graph~$G$ as follows.
		First take the disjoint union of~$H$ and~$H'$.
		Next, consider any $u \in V(H)$ and $v \in (H')$ with colors $c_u = \col_H(u)$ and $c_v = \col_{H'}(v)$.
		Then for each symbol $r \in R_2$, add the edge~$(r,u,v)$ if and only $(r,\rightdir,c_u,c_v) \in M$,
		and symmetrically add the edge $(r,v,u)$ if and only if $(r,\leftdir,c_v,c_u) \in M$.
        \end{itemize}
        
        A \emph{clique-decomposition} $\mathcal{C}_G$ of a $C$-colored graph $G$ is a binary 
        tree such that every node is labeled by one of these operations, which describes how to construct 
        the graph~$G$ using only colors in $C$. 
        There are $|G|$ nodes labeled as constants, which correspond to the leaves of 
	$\mathcal{C}_G$, also denoted $|\mathcal{C}_G|$;
	all the nodes labeled by $\recol$ are of arity $1$
	and all the nodes labeled by $\join$ are of arity $2$.
	We extend the notion of clique-decomposition to non-colored graphs $G$,
        saying that $\mathcal{C}_G$ is a clique-decomposition of $G$
        if there exists a coloration $\col_G$ of $G$ such that 
        $\mathcal{C}_G$ is a clique-decomposition of $(V(G), E(G), \col_G)$.
        
        \begin{remark}
        	\label{r:sameGtoCD}
		One clique-decomposition constructs exactly one $C$-colored graph which correspond to one non-colored
        	graph, so we will often identify a clique-decomposition with the unique graph it constructs. 
        	On the other hand, one graph (colored or not) may have multiple clique-decompositions.
        \end{remark}
        
        The \emph{width} of a clique-decomposition $\mathcal{C}$ is the number of colors it effectively uses,
	i.e.\ the number of $c \in C$ such that at some node the constructed graph has at least one vertex of color $c$.
        A graph $G$ has \emph{clique-width} $k$ if there exists a clique-decomposition
        of~$G$ of width $k$, and no clique-decomposition of $G$ with fewer colors.
	In other words, it is the minimum number of colors needed to construct a colored graph
        isomorphic to $G$ with the operations.
	This variant of clique-width is called NLC-width, it is smaller and at least half of the actual clique-width
	for any graph~\cite{johanson98nlcwidth,johansson2001graph}, hence for the purpose of the present article
	(considering families of bounded clique-width) it is equivalent.
        Clique-width is a graph-parameter more general than
        tree-width, since every family of graphs of bounded tree-width is also of
	bounded clique-width (the converse is in general false)~\cite{corneil2001relationship}.
	An example is given on Figure~\ref{fig:cliquedecomposition}.
        
        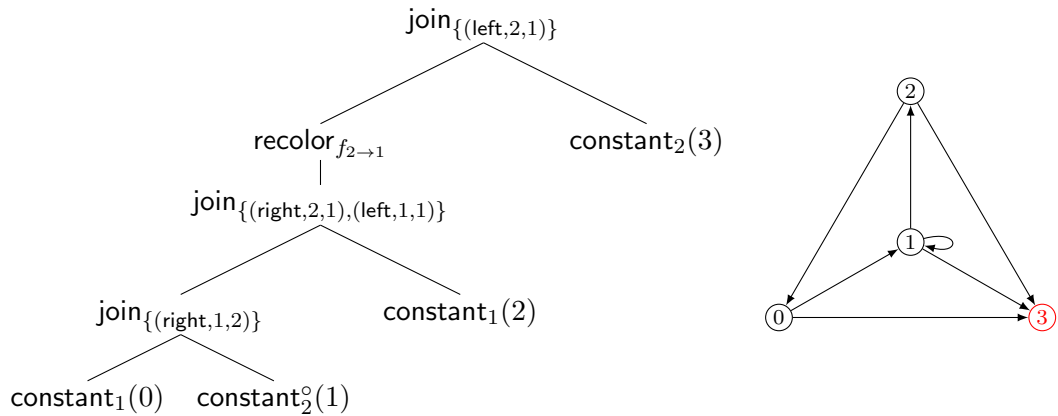
\begin{figure}
        	\begin{minipage}[c]{0.6\linewidth}
        		\centering
        		\Tree [.$\join_{\{(\leftdir, 2, 1)\}}$
        		[.$\recol_{f_{2 \to 1}}$
        		[.$\join_{\{(\rightdir, 2, 1), (\leftdir, 1, 1)\}}$
        		[.$\join_{\{(\rightdir, 1, 2)\}}$
        		$\const_1(0)$
        		$\const_2^{\circ}(1)$
        		]
        		$\const_1(2)$
        		]
        		] 
        		$\const_2(3)$
        		]
        	\end{minipage}
        	\begin{minipage}[c]{0.4\linewidth}
        		\centering
        		\begin{tikzpicture}
        			\tikzstyle{graphe} = [draw,circle,inner sep=0pt,minimum size=10pt]
        			\tikzstyle{arc} = [-latex]
        			
        			\node[graphe] (v0) at (210:2cm) {\scriptsize $0$};
        			\node[graphe] (v1) at (0:0cm) {\scriptsize $1$};
        			\node[graphe] (v2) at (90:2cm) {\scriptsize $2$};
        			\node[graphe, red] (v3) at (330:2cm) {\scriptsize $3$};

        			\draw[arc] (v0) to (v1);
        			\draw[arc] (v2) to (v0);
        			\draw[arc] (v1) to (v2);
        			\draw[arc] (v1) to[loop right] (v1);
        			\draw[arc] (v0) to (v3);
        			\draw[arc] (v1) to (v3);
        			\draw[arc] (v2) to (v3);
        		\end{tikzpicture}
        	\end{minipage}
		\caption{
			Clique-decomposition (left) of a graph (right).
			Color $1$ is in black and color $2$ is in red.
        		We denote $f_{2\to 1}$ the function such that
        		$f_{2\to 1}(2) = 1$ and $f_{2\to 1}(i) = i$ for any $i \ne 1$.
        		The width of this decomposition is $2$.
        		There is only one arrow label, so we do not mention it in the $\join_M$ 
        		operations, such that we 
        		can consider $M \subseteq \{\leftdir,\rightdir\} \times C^2$, and neither 
        		in the $\const^{\circ}$ operation where we use $\circ$ to mention the
        		loop with the only considered arrow label.
		}
        	\label{fig:cliquedecomposition}
        \end{figure}

	When considering the tree of a clique-decomposition of~$G$,
	the $\join$ nodes are the only ones with degree~2, and the leaves are in bijection with the vertices of~$G$.
	This immediately implies the following.
	\begin{remark}\label{rmk:decomposition-depth}
		If~$\Cc$ is a clique decomposition of a graph~$G$ with more than $2^{h-1}$ vertices,
		then there is a branch of~$\Cc$ containing at least~$h$ $\join$ nodes.
	\end{remark}

        \subsection{MSO and graph logics}\label{ss:mso}
        
        We consider signatures of the form $\mathcal{S} = \{=\} \cup R_2$ where 
        $=$ is the equality of vertices, and~$R_2$ is a set of binary relations over the vertices (for example $\to$, $\le$, 
        \emph{etc.}) that we also see as labels on arrows.
	\emph{Monadic Second Order} (MSO) logics allows to quantify over vertices and sets of vertices,
	see for example~\cite{libkin2004elements} (also for more on types).
        The quantifier rank of a formula $\psi$ is its depth of quantifier nesting.
        Given an MSO formula $\psi$, we say that a graph $G$ is a model of $\psi$ and write 
        $G \vDash \psi$ when the graphical property described by $\psi$ is true in $G$. If it is not, 
        we say that $G$ is a counter-model and write $G \nvDash \psi$.

        For a graph family~$F$, the formula~$\psi$ is \emph{$F$-non-trivial} if it has 
        infinitely many models and infinitely many counter-models in~$F$. Otherwise, 
        $\psi$ is \emph{$F$-trivial}.
	We say that $\psi$ is \emph{cw-non-trivial} 
        if there exists $k$ such that $\psi$ has infinitely many models of clique-width at most $k$ and infinitely many counter-models 
        of clique-width at most $k$. Otherwise, $\psi$ is \emph{cw-trivial}.
        Given MSO formulae $\psi$ and $\chi$, we say that $\psi$ is \emph{cw-non-trivial under the restriction $\chi$}, 
        if $\psi$ has infinitely many models of clique-width at most $k$ that are also models of $\chi$ and 
        infinitely many counter-models of clique-width at most $k$ that are also models of $\chi$. 
        Otherwise, $\psi$ is said to be \emph{cw-trivial under the restriction $\chi$}.
        Formula $\chi$ is seen as a \emph{restriction} over admissible structures, and
	for short we may simply write that~$\psi$ is trivial or non-trivial under $\chi$.
	
	We say that a formula $\chi$ is \emph{union-stable} if its set of models is closed under disjoint 
	union (denoted $\sqcup$), and \emph{union-unstable} otherwise.
	For example,
	the formula $\chi = \forall x, \exists y,(x \to y) \wedge (\forall z, \neg(z = y) \implies \neg(x \to z))$,
	meaning that $\to$ is a deterministic relation,
	is union-stable because the union of two graphs of out-degree exactly $1$
	is still a graph of out-degree exactly $1$.
	However,
	the formula $\chi = \forall x, y, z,( x \le x) \wedge (x \le y \wedge y \le x \implies x = y)
	\wedge (x \le y \wedge y \le z \implies x \le z) \wedge ( x \le y \vee y \le x)$,
	meaning that $\le$ is a total order, is union-unstable because models are necessarily connected for this 
	relation, hence for non-empty $G$ and $H$ we always have $G \sqcup H \nvDash \chi$.
	Let us insist on the fact that
	the notion of union-stability is relative to $\chi$ only,
	and not to the fact that formula $\psi$ is trivial or not.
	In fact, union-stability is useful when we want to construct models or counter-models for $\psi$
	still under the restriction $\chi$.
		
		
        \begin{definition}\label{def:spectrum}
        	The \emph{spectrum} of a formula $\psi$ is the set of sizes of 
		finite models of $\psi$, \emph{i.e.}~%
        	$\spec(\psi) = \{|G| : G \vDash \psi, |G| < \infty \}$.
        \end{definition}

        We say that a couple $(\psi,\chi)$ is \emph{$F$-size-independent} when
        the spectrum of $\psi\wedge\chi$ and of $\neg\psi\wedge\chi$
        have an infinite intersection among graphs in the family $F$.
        Otherwise it is called $F$-size-dependent, and consequently under the restriction $\chi$
        the model-checking of~$\psi$ reduces to a consideration on the size of structures.
        Finally, $(\psi,\chi)$ is \emph{cw-size-independent} if there is a~$k$ such that
	for~$F$ the family of graphs of clique-width at most~$k$, $(\psi,\chi)$ is $F$-size-independent.
        That is, there are arbitrarily large graph sizes $n$ such that:
        \begin{itemize}
          \item there exists $G$ of size $n$ and clique-width at most $k$ such that
            $G\models\psi$ and $G\models\chi$;
          \item there exists $G'$ of size $n$ and clique-width at most $k$ such that
            $G'\not\models\psi$ and $G\models\chi$.
        \end{itemize}

	Let us point out that \cref{lemma:pumpingUnstable} implies that
	if~$\psi$ is cw-non-trivial under~$\chi$, and~$\chi$ is union-stable,
	then $(\psi,\chi)$ is cw-size-independent.
	Hence, cw-size-independence is only used in cases where~$\chi$ is not assumed to be union-stable.

        \subsection{MSO types for clique-decompositions}
	Let $\mathcal S=\{=\}\cup R_2$ be a signature, $C$ a set of colors, and $G$ a $C$-colored graph.
	When considering MSO logic for~$G$, we also allow formulae to refer to the color of vertices in~$C$,
	i.e.\ formally we consider formulae with signature $\mathcal{S} \cup C$ with colors in~$C$ being unary predicates.
        \begin{remark}
		As in Remark~\ref{r:sameGtoCD}, if~$\Cc$ is a clique-decomposition corresponding to the $C$-colored graph~$G$,
		then we use $\Cc \vDash \psi$ to denote that $G \vDash \psi$.
        \end{remark}

        The \emph{quantifier-rank-$m$ MSO type of $G$}
        is the set of MSO formulae of quantifier rank at most $m$ that are satisfied by $G$,
	that is $t_m(G) = \{ \psi \mid G \vDash \psi, r(\psi) \le m\}$ where $r(\psi)$ is the quantifier rank of $\psi$.
        We say that a set of formulae of quantifier rank at most~$m$ is a 
        \emph{realized type of quantifier rank $m$} if it is the type of some graph~$G$.
	When~$G$ is a $C$-colored graph, its type~$t_m(G)$ is defined using all MSO over the signature~$\mathcal{S} \cup C$.

	The following standard result is essential for algorithmic aspects of model-checking,
	and can be proved by reducing any formula of rank~$m$ to a normal form with size bounded by a function of~$m$.
        \begin{lemma}[{\cite[Proposition~7.5]{libkin2004elements}}]\label{lemma:finitely-many-types}
		For any finite signature~$\mathcal S$ and rank $m \in \N$,
        	there are only finitely many types of quantifier rank $m$.
        \end{lemma}
	The signature~$R_2$ being implicitly fixed and finite, for a set~$C$ of colors,
	we denote by~$T^C_m$ the (finite) set of all types of $C$-colored graphs of quantifier rank~$m$.
        
         \begin{remark}
		Again, we extend the notion of type to clique-decompositions as explained in Remark~\ref{r:sameGtoCD},
		that is if $G$ has a clique-decomposition $\mathcal{C}$ then we denote
		$t_m(\mathcal{C})=t_m(G)$.
        \end{remark}
        
	Now we introduce a gluing operation on clique-decompositions, analogous to that on tree-decompositions
	used in~\cite{ggpt21}, and for this purpose we introduce a placeholder in a clique-decomposition
        as in~\cite{gglgopt25}.
	A \emph{marked} clique-decomposition is a clique-decomposition with a unique special leaf node labeled by~$\square$.
	A marked clique-decomposition does not directly give a graph;
	rather it is used as a building piece for larger clique-decompositions.
        
	The operation of \emph{gluing} two clique-decompositions $\mathcal{C} \glue \mathcal{C'}$, with 
        $\mathcal{C}$ being a marked clique-decomposition, consists in replacing
        the $\square$ node of $\mathcal{C}$ by $\mathcal{C'}$.
	Thus $\Cc \glue \Cc'$ is marked if and only if~$\Cc'$ is marked.
	The gluing $\glue$ is not commutative but it is associative,
	and we introduce a notation for sequences of gluings.
	Given a set of clique-decompositions $\Gamma = \{\mathcal{C}_i\}_{i \in I}$ for a finite set~$I$,
        and a non-empty finite word $w= w_1 \cdots w_{m} \in I^{m}$,
        we define~$\Delta^{\Gamma}(w)$ inductively:
        \[
        \Delta^{\Gamma}(w_1) = \mathcal{C}_{w_1}
        \quad\text{and}\quad 
        \Delta^{\Gamma}(w_1 \cdots w_{m}) = 
        \Delta^{\Gamma}(w_1 \cdots w_{m-1}) \glue \mathcal{C}_{w_{m}}.
        \]
	Implicitly, this notation requires all~$\Cc_i$ except possibly~$\Cc_m$ to be marked.
        
        %

	Gluing clique-decompositions is compositional relative to MSO types, as expressed in the following lemma.
        \begin{lemma}[{\cite[Corollary~5.60]{CourcelleE2012}}]
        	\label{lem:compositionality-oplus}
		For a set~$C$ of colors and $m \in \N$, consider a marked clique-decomposition~$\C_1$,
		and two $\C_2,\C'_2$ two clique decompositions with the same quantifier-rank-$m$ MSO type (as $C$-colored graphs).
		Then $\C_1 \glue \C_2$ and $\C_1 \glue \C'_2$ also have the same quantifier-rank-$m$ MSO type.
        \end{lemma}
	It is crucial for this statement that the MSO types of $C$-colored graphs
	are defined with MSO formulae that may refer to the colors of vertices.

	A special case of this result is for disjoint unions, obtained by applying \cref{lem:compositionality-oplus} to the two sides of a $\join_{\emptyset}$ node.
	\begin{corollary}\label{cor:compositionality-disjoint-union}
		For any $m \in \N$, the type $t_m(G \sqcup G')$ only depends on $t_m(G)$ and $t_m(G')$.
	\end{corollary}

	\Cref{lem:compositionality-oplus} informally says that clique-decompositions of the same type are interchangeable and,
	given that there are finitely many types (Lemma~\ref{lemma:finitely-many-types}),
	we can use this tool in big enough models to apply pumping techniques.
	To this purpose it is convenient to see a clique-decomposition as a tree, where
        the label of each node is the type of the graph generated by the sub-clique-decomposition
	consisting of this node and its descendants.

	Given a marked clique-decomposition $\C$, let us denote by $f_\C$ the function obtained as
	the composition of the recoloring functions $f$ of all $\recol_f$ nodes in
	the path from the marked leaf~$\square$ to the root of $\C$.
	We say that $\C$ is \emph{idempotent} when $f_{\C} \circ f_\C =f_{\C}$.
	This property will be important when constructing circuits.

	The following is a folklore result, which holds in any finite semigroup.
	\begin{lemma}[Idempotent power lemma]\label{lem:idempotent-power}
		For any map $f: [k] \to [k]$, there exists some~$n>0$ such that~$f^n$ is idempotent.
	\end{lemma}
	\begin{proof}
		By pigeonhole principle, there are $n,i$  such that $f^n = f^{n+i}$ with $n,i>0$ and $n+i \le k!+1$.
		Multiplying both side by~$f^{i\ell}$, we also have $f^{n+i\ell} = f^{n+i(\ell+1)}$,
		and thus $f^n = f^{n+i\ell}$ for any $\ell \in \N$.
		Now
		\[ f^{ni} = f^{n(i-1)} \circ f^n = f^{n(i-1)} \circ f^{n(i+1)} = f^{2ni}, \]
		proving that~$f^{ni}$ is idempotent.
	\end{proof}

  \subsection{Problems \texorpdfstring{$\psi$}{ψ}-under-\texorpdfstring{$\chi$}{χ}-dynamics}

	Our goal is to study the computational complexity of decision problems of the form:
	given a succinct representation of an $R_2$-graph~$G$ satisfying some formula $\chi$,
        test whether~$G$ satisfies formula $\psi$.
	The condition~$\chi$ is treated as a promise, in order to avoid considering the complexity
        of checking the input validity (even if it is actually not an issue for lower bounds).
        It is meant to enforce structural properties of the graphs under consideration, and to enforce the interpretation of
        additional relations in $R_2$.
        The main question is asked by $\psi$.

        \medskip
        \noindent
        \centerline{
        \decisionpbpromisew{$\psi$-under-$\chi$-dynamics}
        {a succinct $R_2$-graph representation $(N,D)$.}
        {$\dyna{N,D}\models \chi$.}
        {does $\dyna{N,D}\models \psi$ ?}
        {.7}
        }

       
       

        \section{Pumping on clique decompositions}
        \label{s:pumpings}
        
        In this section, we construct building blocks (pieces of clique-decompositions)
	which will be used in the reductions to prove complexity lower bounds in Section~\ref{s:reductions}.
	We begin with a general pumping lemma in Subsection~\ref{ss:pumping}.
	Then, we state two more specialized lemmas, used to prove \cref{thm:main-NP-hard,thm:main-P-hard} respectively.
	These lemmas give some constructions through pumping of both models and counter-models of the MSO formula~$\psi$ of interest.
	It is particularly important to keep in these results the ability to construct both models and counter-models of~$\psi$ of size~$n$ for infinitely many~$n$,
	i.e.\ a form of size-independence in the conclusion of the statements.

	Under union-stable restrictions (\cref{thm:main-NP-hard}),
	we can adapt the construction through saturating graphs used in~\cite{gglgopt25} (Subsections~\ref{ss:unionstable}).
	Without union-stability, we need to apply pumping to models and to counter-models in parallel (Subsection~\ref{ss:unionunstable}).
	The cw-size-independence assumption is crucial in the latter case,
	whereas it is implied by union-stability in the former one.

        \subsection{Construction of models by pumping}
        \label{ss:pumping}
	The next result is a generic pumping lemma, from which we will derive the more specific statements needed for the latter sections.
	Here, we work with graphs of clique-width~$k$, constructed with the color set~$[k]$.
	Recall that $T^{[k]}_m$ denote the finite set of MSO types of $[k]$-colored graphs of quantifier rank $m$.
	Recall also that if~$\Gamma$ is a family of (marked) clique-decompositions indexed by~$I$,
	and~$w$ is a sequence of indices in~$I$, then $\Delta^\Gamma(w)$ is obtained
	by gluing together the corresponding sequence of clique-decompositions from~$\Gamma$.

        \begin{lemma}
       	\label{lemma:pumpingModels}
	Let~$\phi$ be an MSO formula with quantifier rank~$m$.
	Suppose that~$\phi$ has a model with clique-width at most~$k$ and size $x > 2^{|T^{[k]}_m|}$.
       	Then there exist three clique-decompositions $\Gamma = \{\C_s, \C_r, \C_e\}$ (for `start', `repeat', and `end') of width at most~$k$, with~$\C_s,\C_r$ marked, such that:
	\begin{itemize}
		\item for any $\ell \in \N$, $\Delta^\Gamma(s r^{\ell} e) \vDash \phi$,
		\item $|\C_s \glue \C_e| = x$, and
		\item $\C_r$ is not empty.
	\end{itemize}
        \end{lemma}
      	\begin{proof}
      		Let $\Cc$ be a clique-decomposition of with~$k$ of a model of $\phi$ of size $x$.
		For each node~$v$ in~$\Cc$, consider the sub-decomposition~$\Cc_v$ corresponding to the subtree below~$v$.
		By \cref{rmk:decomposition-depth}, since $x > s^{|T^{[k]}_m|}$,
		there is a branch of~$\Cc$ containing at least~$|T^{[k]}_m|+r$ $\join$ nodes.
		Now by pigeonhole principle, there are two $\join$ nodes $v \neq v'$ in this branch
		with the same type $t_m(\Cc_v) = t_m(\Cc_{v'})$,
		say with~$v$ ancestor of~$v'$ in the tree.
      		
		Define now $\C_s, \C_r, \C_e$ as follows:
      		\begin{itemize}
      			\item $\C_s$ is the marked clique-decomposition obtained from~$\C$ by replacing~$v$ and all its subtree by the marked node~$\square$;
      			\item $\C_e = \Cc_v$ is the (unmarked) clique-decomposition consisting of the descendants of~$v$ and~$v$ itself;
      			\item $\C_r$ is the marked clique-decomposition obtained from~$\C_e$ by replacing~$v'$ and all its subtree by the marked node~$\square$.
      		\end{itemize}
		The definition directly gives, $\C_s \glue \C_e = \C$, hence it has exactly $x$ vertices.
		Moreover, since $v \ne v'$ and~$v$ is a $\join$ node, $\C_r$ contains at least one constant node, hence it is not empty.

		Consider the type $\tau = t_m(\C_v) = t_m(\C_{v'})$.
		The construction gives $\C_e = \C_v = \C_r \glue \C_{v'}$.
		Thus, by compositionality (\cref{lem:compositionality-oplus}),
		for any clique-decomposition~$\C'$, if $t_m(\C') = \tau$, then $t_m(\C_r \glue \C') = \tau$.
		Starting with~$\C_e$ and iteratively applying this remark gives that
		$t_m(\Delta^\Gamma(r^\ell e)) = t_m(\C_r \glue \dots \glue \C_r \glue \C_e) = \tau$.
		A final application of \cref{lem:compositionality-oplus} using~$\C_s$ then gives that
		$\Delta^\Gamma(s r^\ell e)$ has the same rank-$m$ type as $\C_s \glue \C_e = \C$.
		Thus, since the latter satisfies~$\phi$ which has rank~$m$, so does the former,
		i.e.\ $\Delta^\Gamma(s r^\ell e) \vDash \phi$ as desired.
      	\end{proof}
       
        \subsection{Union-stable \texorpdfstring{$\chi$}{χ}}
        \label{ss:unionstable}

        In this case, we can use the disjoint union to construct a saturating graph of rank $m$,
	which will become $\C_b$ in Lemma~\ref{lemma:pumpingStable}.
	The three other clique-decompositions will be obtained from Lemma~\ref{lemma:pumpingModels}.
        
        \paragraph*{Saturating graph}

	In~\cite{gglgopt25}, it is shown that for all $m \in \N$, there exists a \emph{saturating graph} $\Omega_m$ such that, 
        for any MSO formula $\psi$ of quantifier rank $m$, if a graph contains a copy of $\Omega_m$ as connected component, then either it is forced to be a model of $\psi$
	or it is forced to be a counter-model of $\psi$.
	The proof consists in constructing $\Omega_m$ as the disjoint union of many copies of graphs for each MSO types.
	The reasoning works identically for other (finite) signatures $\mathcal{S}$,
	when adding a union-stable restriction~$\chi$, and when constraining the graphs to have bounded clique-width.

        \begin{proposition}[variant of~{\cite[Proposition~4.1]{gglgopt25}}]
	\label{prop:saturating}
	Fix $k,m \in \N$ and a signature $\mathcal{S}$. Let $\chi$ be a union-stable MSO restriction.
       	There exists a graph $\Omega_{\mathcal{S}, m, k,\chi}$ such that, for every
       	MSO formula $\psi$ of rank $m$, either:
       	\begin{itemize}
       		\item for any graph $G$ of clique-width at most~$k$ such that $G \vDash \chi$, we have $G \sqcup \Omega_{\mathcal{S}, m, k,\chi} \vDash \psi$, or
       		\item for any graph $G$ of clique-width at most~$k$ such that $G \vDash \chi$, we have $G \sqcup \Omega_{\mathcal{S}, m, k,\chi} \nvDash \psi$,
       	\end{itemize}
       	and moreover
       	$\Omega_{\mathcal{S}, m, k,\chi}$ satisfies~$\chi$, and has clique-width at most~$k$.
	\end{proposition}
	\begin{proof}
	  Let~$\tau_1,\dots,\tau_n$ be all the MSO types of rank~$m$ such that
	  there is a graph~$G_i$ with type~$\tau_i$, clique-width at most~$k$, and such that $G_i \vDash \chi$.
	  Note that there are only finitely many of them by \cref{lemma:finitely-many-types}.

	  Now by \cite[Lemma~4.2]{gglgopt25}, for each~$G_i$ (and indeed any graph),
	  there is some $N_i \in \N$ such that MSO formulae of quantifier rank~$m$
	  cannot distinguish the number of copies of~$G_i$ when there are at least~$N_i$ of them:
	  formally, for any $\ell,\ell' \ge N_i$, the types $t_m\left(\bigsqcup_{j=1}^\ell G_i\right)$ and  $t_m\left(\bigsqcup_{j=1}^{\ell'} G_i\right)$ are the same.

	  Now define $\Omega_{\mathcal{S},m,k,\chi} = \bigsqcup_{i=1}^n \bigsqcup_{j=1}^{N_i} G_i$.
	  First, since each connected component is some~$G_i$ which satisfies~$\chi$ and has clique-width at most~$k$, so does the whole $\Omega_{\mathcal{S},m,k,\chi}$.
	  Next, consider~$G$ a graph with clique-width at most~$k$ satisfying~$\chi$.
	  By choice of the types $\tau_1,\dots,\tau_m$, we have $t_m(G) = \tau_i$ for some~$i$.
	  Then we have
	  \[ t_m\left(G \sqcup \bigsqcup_{j=1}^{N_i} G_i\right)
	  = t_m\left(\bigsqcup_{j=1}^{N_i+1} G_i\right)
	  = t_m\left(\bigsqcup_{j=1}^{N_i} G_i\right), \]
	  where the first equality uses compositionality (\cref{cor:compositionality-disjoint-union})
	  to replace~$G$ with yet another copy of~$G_i$, and the second is by choice of~$N_i$.
	  Adding the rest of~$\Omega_{\mathcal{S},m,k,\chi}$ to this equality using again \cref{cor:compositionality-disjoint-union},
	  we find $t_m(G \sqcup \Omega_{\mathcal{S},m,k,\chi}) = t_m(\Omega_{\mathcal{S},m,k,\chi})$,
	  which implies the result.
	\end{proof}
       
       
       
	We now have all the tools for the following lemma,
	which will be used in hardness proofs under union-stable restrictions.
	The proof is similar to~\cite{gglgopt25,glp25v2}.
        \begin{lemma}
       	\label{lemma:pumpingStable}
       	Let $\psi$ and $\chi$ be MSO formulae over signature 
       	$\mathcal{S} = \{=\} \cup R_2$ 
       	with binary relations $R_2$. 
	If $\psi$ is cw-non-trivial 
       	under the restriction $\chi$ and $\chi$ is \textbf{union-stable},
       	then there exist $k \in \N$ and 
       	four clique-decompositions $\Gamma = \{\C_s, \C_g, \C_b, \C_e\}$
	(for `start', `good', `bad', `end')
       	of width at most $k+1$, with~$\C_s,\C_g,\C_b$ marked,
       	such that the following holds either for $\psi'=\psi$ or for $\psi'=\neg\psi$:
	\begin{itemize}
	  \item for any word $w \in \{g,b\}^\star$, we have $\Delta^\Gamma(s w e) \vDash \chi$,
	  \item for any word $w \in \{g,b\}^\star$, we have $\Delta^\Gamma(s w e) \vDash \psi'$ if and only if~$w$ consists only of~`$g$'s,
	  \item $\C_g$ and~$\C_b$ have the same number of vertices, and are non-empty, and
	  \item the recoloring maps $f_{\C_g}$ and~$f_{\C_b}$ are equal, and are idempotent.
        \end{itemize}
        \end{lemma}
        \begin{proof}
	Choose~$m$ to be the maximum quantifier rank of~$\psi$ and~$\chi$.
	We assume that $\Omega = \Omega_{\mathcal{S}, m,k \chi}$ given by Proposition~\ref{prop:saturating} satisfies
        $\Omega \nvDash \psi$ and we prove the Lemma for $\psi'=\psi$,
        otherwise the proof is identical for $\psi'=\neg\psi$.

	By assumption, $\psi$ is cw-non-trivial under~$\chi$, hence there is some $k \in \N$ such that
	$\psi \land \chi$ has infinitely many models of clique-width at most~$k$.
	In particular, there is a model with size more than~$2^{|T_m^{[k]}|}$, allowing to apply \cref{lemma:pumpingModels}.
	This gives three clique-decompositions $\tilde{\Gamma} = \{\C_s, \C_r, \C_e\}$ of width at most $k$
	such that $\Delta^{\tilde{\Gamma}}(s r^{\ell} e) \vDash \psi \land \chi$ for any $\ell \in \N$, and with~$\C_r$ non-empty.

	\Cref{lem:idempotent-power} applied to the recoloring map~$f_{\C_r}$
	gives some~$t$ such that $f_{\C_r}^t$ is idempotent.
	Thus, $\C'_r = \glue_{i=1}^t \C_r$ (i.e.\ $\C_r$ glued to itself~$t$ times) is idempotent.
	Next, define $\C''_r = \glue_{i=1}^{|\Omega|} \C'_r$ and $\Omega' = \bigsqcup_{i=1}^{|\C'_r|} \Omega$.
	This is to ensure $|\C''_r| = |\Omega'|$.
	Since~$\C''_r$ is just several copies of~$\C_r$ glued together,
	we still have that $\C_s \glue \C''_r \glue \dots \glue \C''_r \glue \C_e \models \psi \land \chi$.
	Similarly, since~$\Omega'$ is just several disjoint copies of~$\Omega$,
	we have $\Omega' \vDash \chi$, and $\Omega' \sqcup \mathcal{D} \nvDash \psi$ whenever $\mathcal{D} \vDash \chi$.

	Let us now define the marked clique-decompositions~$\C_g$ and~$\C_b$.
	Note that we have so far only worked with clique-decompositions of width~$k$,
	say using colors $\{0,\dots,k-1\}$, but the statement allows one additional color~$k$.
	\begin{itemize}
		\item $\C_g$ is defined as $\C''_r \glue \C''_r$.
		\item $\C_b$ is constructed by taking a clique-decomposition of width~$k$ for~$\Omega'$,
		recoloring all vertices of~$\Omega'$ with the unused color~$k$,
		then taking the disjoint union (with a $\join_\emptyset$ node) with~$\C''_r$.
	\end{itemize}

	We finally define the family~$\Gamma = \{\C_s,\C_g,\C_b,\C_e\}$.
	Since $|\C''_r| = |\Omega'|$, we clearly have $|\C_g| = |\C_b|$.
	The construction of~$\C_b$ only adds~$\Omega'$ to the side of~$\C''_r$, without recoloring the latter.
	Thus the recoloring map of~$\C_b$ is the same as that of~$\C''_r$, which is also the same as~$\C_g = \C''_r \glue \C''_r$ by idempotence.

	Finally, for a word $w \in {g,b}^*$, consider $\Delta^\Gamma(s w e)$.
	The effect of each~$g$ in~$w$ is to add two copies of~$\C''_r$ in the gluing sequence,
	while the effect of each~$b$ is to add one copy of~$\C''_r$ in the gluing,
	and one copy of~$\Omega'$ disjoint from the whole graph.
	Thus, if~$i,j$ are the number of occurrences of~$g,b$ respectively in~$w$,
	\[ \Delta^\Gamma(s w e) = 
	(\C_s \glue \underbrace{\C''_r \glue \dots \glue \C''_r}_{\text{$2i+j$ times}} \glue \C_e) \sqcup \bigsqcup_{i=1}^j \Omega'. \]
	If~$w$ contains only the letter~$g$ (i.e.\ $j=0$), this is just
	$\C_s \glue \C''_r \glue \dots \glue \C''_r \glue \C_e$,
	which we know satisfies $\psi \land \chi$ as desired.
	If however $j>0$, we still have that $\Delta^\Gamma(s w e)$ satisfies~$\chi$, because each connected component does,
	but it does not satisfy~$\psi$ by choice of the saturating graph~$\Omega$,
	proving the result.
        \end{proof}
        
	\subsection{Cw-size-independent \texorpdfstring{$(\psi,\chi)$}{(ψ,χ)}}
        \label{ss:unionunstable}
        

        \begin{lemma}
       	\label{lemma:pumpingUnstable}
       	Let $\psi$ and $\chi$ be MSO formulae over signature 
       	$\mathcal{S} = \{=\} \cup R_2$ 
       	with binary relations $R_2$, such that
	$\psi$ is cw-non-trivial under the restriction $\chi$
	and $(\psi,\chi)$ is \textbf{cw-size-independent}.
	Then there exist $k \in \N$ and six clique-decompositions $\Gamma = \{\C_{s_+}, \C_{r_+}, \C_{e_+}, \C_{s_-}, \C_{r_-}, \C_{e_-}\}$ 
       	of width at most $k$, with $\C_{s_+}, \C_{r_+}, \C_{s_-}, \C_{r_-}$ marked, such that the following holds.
       	\begin{itemize}
       		\item $\forall \ell\in\N$, $\Delta^\Gamma(s_+  r_+^{\ell}  e_+) \vDash \psi \land \chi$,
       		\item $\forall \ell\in\N$, $\Delta^\Gamma(s_-  r_-^{\ell}  e_-) \vDash \lnot \psi \land \chi$, and
       		\item $\forall \ell\in\N$, $|\Delta^\Gamma(s_+  r_+^{\ell}  e_+)| = |\Delta^\Gamma(s_-  r_-^{\ell}  e_-)|$.
       	\end{itemize}
	Moreover $\C_{r_+}$ and $\C_{r_-}$ are not empty (they have at least one constant node).
       	\end{lemma}
        \begin{proof}
	By the cw-size-independence assumption, there is some bound $k \in \N$
	such that there are infinitely many sizes $x \in \N$ and clique-decompositions $\C,\C'$ of size~$x$ and width at most~$k$ such that $\C \vDash \psi \land \chi$ and $\C' \vDash \lnot \psi \land \chi$.

	Let~$m$ be the maximum quantifier rank among~$\psi$ and~$\chi$, and
	choose such a size~$x$ of models with $x > 2^{|T_m^{[k]}}$.
       	By Lemma~\ref{lemma:pumpingModels} applied to both $\psi \land \chi$ and $\neg\psi \land \chi$,
	we obtain two triples of clique-decompositions $\{\C_{s_+}, \tilde{\C}_{r_+}, \C_{e_+}\}$
	and $\{\C_{s_-}, \tilde{\C}_{r_-}, \C_{e_-}\}$ of width at most $k$ that satisfy
	the two first items of Lemma~\ref{lemma:pumpingUnstable},
	and with $|\C_{s_+}\glue\C_{e_+}|=|\C_{s_-}\glue\C_{e_-}|=x$.
	Let $\C_{r_+} = \glue_{|\tilde{\C}_{r_-}|} \tilde{\C}_{r_+}$ be the gluing of $\tilde{\C}_{r_+}$
	to itself $|\tilde{\C}_{r_-}|$ time, and symmetrically
	let  $\C_{r_-} = \glue_{|\tilde{\C}_{r_+}|} \tilde{\C}_{r_-}$.
       	The clique-decompositions $\Gamma = \{\C_{s_+}, \C_{r_+}, \C_{e_+}, \C_{s_-}, \C_{r_-}, \C_{e_-}\}$
	still satisfy the two first items of Lemma~\ref{lemma:pumpingUnstable},
	and furthermore we now have $|\C_{r_+}| = |\C_{r_-}|$, from which we deduce the third item.
	Since the clique-decompositions $\tilde{\C}_{r_+}$ and $\tilde{\C}_{r_-}$ are not empty,
	neither are $\C_{r_+}$ nor $\C_{r_-}$.
        \end{proof}
       
        \section{Reductions and complexity lower bounds}
        \label{s:reductions}
      	
      	We present two different kinds of complexity results depending on whether the restriction is
        union-stable or cw-size-independent: reduction from \textbf{SAT} or \textbf{UNSAT}
        to obtain $\NP$- or $\coNP$-hardness, or reduction from \textbf{CVP}
        (\emph{Circuit Value Problem}) to obtain $\Poly$-hardness, respectively.
        The latter requires to construct an instance of \textbf{$\psi$-under-$\chi$-dynamics}
        in logarithmic space (see for example~\cite{limits}).
        
        Given $\C$ and $\C'$, the $R_2$-graph constructed by $\C \glue \C'$
        contains the vertices constructed by $\C$ (its constant leaves),
        the vertices constructed by $\C'$ (its constant leaves), and no other.
        However, the situation is different regarding the arcs of the $R_2$-graph constructed by $\C \glue \C'$:
        it contains the labeled arcs constructed by $\C$, the labeled arcs constructed by $\C'$,
        but it may contain additional arcs.
        Indeed, one needs to search in $\C \glue \C'$ whether two leaves are connected,
	by keeping track of their colors while traversing the clique-decomposition until their join
	(least common ancestor).
        The purpose of the next lemma is to implement efficiently the procedure of
        computing \emph{a circuit for the $R_2$-graph} constructed by pumpings.
        This will be particularly important for cw-size-independent restrictions,
        and the idempotency of recoloring functions is crucial.

      	\begin{lemma}
      		\label{lemma:logspace}
                For any fixed triple of clique-decompositions $\Gamma = \{\C_s, \C_r, \C_e\}$
		with $\C_r$ idempotent,
                a circuit for $\Delta^{\Gamma}(s r^{\ell} e)$
		can be constructed in space logarithmic in $\ell$.
      	\end{lemma}
      	\begin{proof}
		In this proof, we denote by $(\mathcal{D}_i)_{i \in \intz{\ell+2}}$ the sequence of clique-decompositions
		in the gluing $\Delta^{\Gamma}(s r^{\ell} e) = \mathcal{D}_0 \glue \mathcal{D}_1 \glue \dots \glue \mathcal{D}_{\ell +1}$,
		that is $\mathcal{D}_0 = \C_s$, $\mathcal{D}_{\ell + 1} = \C_e$ and for any
        	$i \in \intz{\ell}$, $\mathcal{D}_{i+1}$ is the $i$-th copy of $\C_r$. Hence 
		writing $v \in \mathcal{D}_i$ specifies the index of the copy of graph where $v$ belongs to
     		within $\Delta^{\Gamma}(s r^{\ell} e)$.
      	
		The vertices of $\Delta^{\Gamma}(s r^{\ell} e)$ are
		numbered from $0$ to $|\C_s|+\ell\cdot|\C_r|+|\C_e|-1$
		(thus encoded on $\O(\log\ell)$ bits),
		in the order of gluings.
      		Given any vertex $v$ in $\Delta^{\Gamma}(s r^{\ell} e)$,
		function $\mathtt{num\_copy}(v)$ outputs $i\in\intz{\ell+2}$ such that $v\in\mathcal{D}_i$,
		and function $\mathtt{relative}(v,i)$ outputs $v_i\in\intz{|\mathcal{D}_i|}$
		corresponding to $v$ within $\mathcal{D}_i$.
		Computing them requires only basic arithmetics,
		and circuits for them can be produced in space logarithmic in $\ell$,
		as detailed in~\cite{glp25circuit}.
		Now, given two vertices $v$ and $v'$
		with $i=\mathtt{num\_copy}(v)$ and $i'=\mathtt{num\_copy}(v')$,
		we consider two cases to decide whether there is an edge~$(r,v, v')$
		in $\Delta^{\Gamma}(s r^{\ell} e)$, for each $r \in R_2$.
      		\begin{itemize}
			\item If $|i-i'|<2$, then given $\mathtt{relative}(v,i)$ and $\mathtt{relative}(v',i')$
			there are only finitely many different cases to consider,
			which is immediate to implement in a circuit.
			\item If $|i-i'|\geq 2$, then the join of $\mathtt{relative}(v,i)$ and $\mathtt{relative}(v',i')$
			may require to cross multiple copies of $\C_r$,
			but since its recoloring function is idempotent,
			any number of crossings is equivalent (in terms of recoloring) to one.
			Consequently, assuming without loss of generality that $i>i'$,
                        we know that the join appears within $\mathcal{D}_i$, and
			the color of $v'$ at the join is obtained by finite computations
			(before and after the possibly many copies of $\C_r$) and one application
			of the (fixed) function $f_{\C_r}$ in-between.
      		\end{itemize}
      		
      		

		Thus, after computing $\mathtt{num\_copy}(v)=i$, $\mathtt{num\_copy}(v')=i'$,
		as well as $\mathtt{relative}(v,i)$, $\mathtt{relative}(v',i')$ --- which requires
		the circuit to implement arithmetic operations on $\O(\log \ell)$ bits ---
		deciding whether $(v,v') \in r$ requires to get the colors of $v$ and $v'$ at their join,
		which amounts to checking a finite number of cases.
      	\end{proof}
      	
      	We are now able to prove our main two results.

	\mainNP*
      	\begin{proof}
		Under the assumptions of the statement, Lemma~\ref{lemma:pumpingStable} can be applied
		to obtain four clique-decompositions $\Gamma = \{\C_s, \C_g, \C_b, \C_e\}$ satisfying
		the following for either $\psi' = \psi$ or $\psi' = \lnot \psi$:
	for any word $w \in \{g, b\}^*$, $\Delta^{\Gamma}(s w e) \vDash \chi$; and
	$\Delta^{\Gamma}(s w e) \vDash \psi'$ if and only if the letter $b$ does
	not appear in $w$. 
		Let us assume that the former holds for $\psi' = \lnot \psi$,
		and show that $\textbf{SAT}$ reduces to \textbf{$\psi$-under-$\chi$-dynamics}, proving $\NP$-hardness.
		In the other case, we would instead obtain a reduction to \textbf{$\neg\psi$-under-$\chi$-dynamics},
		and thus $\coNP$-hardness for \textbf{$\psi$-under-$\chi$-dynamics}.

                Given $S$ an instance of \textbf{SAT} on $n$ variables,
                the idea is to construct a circuit succinctly representing
                $\Delta^{\Gamma}(s w e)$,
                where $w\in\{g,b\}^{2^n}$ with $w_i=b$ if and only if $i-1$ is the binary encoding of a satisfying
                assignment for $S$. 
                To this end, we adapt the circuit constructed by Lemma~\ref{lemma:logspace}.
		As in the proof of the latter, we denote by $(\mathcal{D}_i)_{i \in \intz{\ell+2}}$
		the sequence of clique-decompositions in the gluing~$\Delta^\Gamma(swe)$.
		Given vertices~$v,v'$, the circuit again computes $\mathtt{num\_copy}(v)=i$, $\mathtt{num\_copy}(v')=i'$,
		as well as $\mathtt{relative}(v,i)$, $\mathtt{relative}(v',i')$.
		Next, it tests whether~$\Dc_i$ and~$\Dc_{i'}$ respectively are~$\Cc_g$ or~$\Cc_b$,
		by testing the \textbf{SAT} instance~$S$ on~$i-1$ and~$i'-1$.
		The argument is now exactly the same as in \cref{lemma:logspace}, with this additional case disjunction on~$\Dc_i,\Dc_{i'}$.
		Precisely, whether or not~$vv'$ is an edge depends only on (1)~the relative order of~$i,i'$,
		(2)~whether $i,i'$ are equal, differ by~1, or differ by at least~2,
		(3)~whether $\Dc_i,\Dc_{i'}$ are~$\Cc_g$ or~$\Cc_b$,
		and (4)~$\mathtt{relative}(v,i)$, $\mathtt{relative}(v',i')$.
		For each of these four elements, there is only a fixed number of choices,
		hence this can be implemented by a constant size circuit.
		The key hypothesis ensuring the above is that the recoloring function $\Dc_{i+1} \glue \dots \glue \Dc_{i'-1}$ (in the case $i' \ge i+2$)
		is always the same, independently of~$i,i'$ and~$w$,
		because each~$\Dc_j$ is either~$\Cc_g$ or~$\Cc_b$, which have the same idempotent recoloring.
                
                This implementation strongly relies on the more detailed logspace reduction
                of \cite{glp25circuit}, but Lemma~\ref{lemma:logspace}
                highlights the transformation from models described
                by tree-decompositions to models described by clique-decompositions.

		We thus obtain in logspace a circuit encoding the graph $G$ generated by $\Delta^{\Gamma}(s w e)$,
		where each letter~$w_i$ is `$b$' if~$i-1$ is a satisfying assignment for the \textbf{SAT} instance~$S$, and is~`$g$' otherwise.
		\Cref{lemma:pumpingStable} guarantees that $G \vDash \chi$,
		and that $G \vDash \psi$ if and only if~$w$ contains a~`$b$',
		i.e.\ if and only if~$S$ is a positive instance of \textbf{SAT},
		completing the reduction to \textbf{$\psi$-under-$\chi$-dynamics}.             
	\end{proof}

	\mainP*
	\begin{proof}
                If $(\psi,\chi)$ is cw-size-independent, then there exist six
                clique-decompositions 
                $\Gamma = \{\C_{s_+}, \C_{r_+}, \C_{e_+}, \C_{s_-}, \C_{r_-}, \C_{e_-}\}$ verifying the conclusion of
	        Lemma~\ref{lemma:pumpingUnstable}. Consider an instance $S$ of 
                \textbf{CVP} (\emph{Circuit Value Problem}), which is a circuit on $n$ variables
                that we want to evaluate on the input $0^n$.
                We construct in logarithmic space a succinct graph representation which is a model of $\psi$
                if $S(0^n)=1$ and a counter-model if $S(0^n)=0$.
                Applying Lemma~\ref{lemma:logspace} to both $\{\C_{s_+},\C_{r_+},\C_{e_+}\}$ and $\{\C_{s_-},\C_{r_-},\C_{e_-}\}$
                with $\ell=|S|$, in space logarithmic in $|S|$ we obtain circuits for
	        $\Delta^{\Gamma}(s_+ r_+^{\ell} e_+)$ and
	        $\Delta^{\Gamma}(s_- r_-^{\ell} e_-)$ which
                both satisfy $\chi$, and are respectively a model and a counter-model of $\psi$.
                Given that they have the same number of vertices (third item of Lemma~\ref{lemma:pumpingUnstable}),
                it suffices to plug them into a multiplexer controlled by $S(0^n)$.
                The only non-constant parts are the dependency on $\ell=|S|$
                (which provides an encoded graph of size proportional to that of its circuit)
                and the copy of circuit $S$
                with the adjunction of constant $0$ for each of its inputs,
                which are easily performed in logarithmic space
                (for more details on the space complexity, see \cite[Section~4]{glp25circuit}).
                We conclude that if $(\psi,\chi)$ is cw-size-independent
                then \textbf{$\psi$-under-$\chi$-dynamics} is $\Poly$-hard.
        \end{proof}

        Note that as the size of the circuit succinctly encoding a graph is required to have size bounded
        polynomially in terms of the size of the graph (in order to avoid hiding the complexity inside the circuit),
        it is necessary to use models and counter-models of arbitrarily large size in the previous reduction, hence the use of Lemma~\ref{lemma:pumpingUnstable}.
      
        \section{Applications and extensions}
        \label{s:applications}

	Restriction $\chi$ has two purposes.
	First, it is meant to control structural properties of the dynamics
	(the succinctly encoded graph), such as being deterministic (each vertex has out-degree one).
	Second, it allows to enforce the interpretation of additional $R_2$ symbols
	as, for example, a total or partial order.
	In this section we discuss some applications,
	and outline the necessity of cw-size-independence.

        \subsection{Example of cw-size-dependent formula}

        In Theorem~\ref{thm:main-P-hard} we introduced the notion of $(\psi,\chi)$ cw-size-independence.
        It would have been better to prove a $\Poly$-hardness result for any union-unstable restriction $\chi$,
        however this does not hold, as demonstrated by the following example.
        Let formula $\chi$ express that the graph is a cycle
        (connected, undirected, each vertex of degree two),
        and formula $\psi$ express that the cycle has even length
        (there is a valid bi-coloring).
	\begin{align*}
	  \chi =&~
            \neg[\exists X, (\exists u, u \in X \wedge \exists v, \neg v \in X) \wedge
	    (\forall u, \forall v, u \to v \Rightarrow (u \in X \Leftrightarrow v \in X))]\\
	    &\wedge [\forall u,\forall v, u \to v \Leftrightarrow v \to u]
	    \wedge [\forall u,\exists v,\exists w,\forall z,\\
            &((u \to v) \wedge (u\to w) \wedge (\neg v = w) \wedge (\neg v = z) \wedge (\neg w = z))
	    \Rightarrow (\neg u \to z)]\\
	  \psi =&~ \exists X, \forall u, \forall v, u \to v \Rightarrow (u \in X \Leftrightarrow \neg v \in X)
	\end{align*}
	This couple is cw-size-dependent, and to decide \textbf{$\psi$-under-$\chi$-dynamics}
	it suffices to check whether the least significant bit of the size $N$ is $0$ or $1$ (very easy).
      	
	\subsection{Example of \texorpdfstring{$\Poly$}{P}-complete problems}

	Many problems are $\Poly$-complete under restriction
	$\xi=(\forall x, \forall y, x \to y) \vee (\forall x, \forall y, \neg x \to y)$.
      	Restriction $\xi$ asserts that either the graph is a clique, or it is an independent set. 
      	
      	\begin{proposition}
      		\label{prop:pcomplete}
		For any cw-non-trivial MSO formula $\psi$ under $\xi$ on signature $\{=,\to\}$,
		problem \textbf{$\psi$-under-$\xi$-dynamics} is $\Poly$-complete
      	\end{proposition}
      	
      	In order to prove this proposition, we prove a much stronger claim
        regarding the models and counter-models of $\psi\wedge\xi$ and their sizes.
      	Let $K_n$ denote the clique of size $n$, and $I_n$ the independent set of
      	size $n$.
      	\begin{lemma}
      		\label{lemma:cliqueindep}
      		Let $\psi$ be a cw-non-trivial MSO formula under $\xi$ on signature $\{=, \to\}$,
      		and let $r$ be the quantifier rank of $\psi$.
      		There exists $N_r \in \N$ such that either:
      		\begin{itemize}
      			\item for any $n > N_r$, we have $K_n \vDash \psi$ and $I_n \nvDash \psi$, or inversely
			\item for any $n > N_r$, we have $I_n \vDash \psi$ and $K_n \nvDash \psi$.
      		\end{itemize}
      	\end{lemma}
      	
        Lemma~\ref{lemma:cliqueindep} means that outside of a finite number of graphs (the graphs with less than 
      	$N_r$ vertices), all the models are cliques or all the models are independent sets.
      	
      	\begin{proof}
      		Towards a contradiction, suppose that for any $N \in \N$, there exist
      		two models $A$ and $B$ of $\xi$, with $|A| > N$ and $|B| > N$, such that $A$ 
      		is a model of $\psi$, $B$ is a counter-model of $\psi$, and that
      		both are cliques or both are independent sets. 
      		For large enough $N$, the following strategy allows Duplicator to win the
      		Ehrenfeucht-Fraïssé game on any formula $\psi$ of rank $r$,
                therefore $\psi$ cannot distinguish these graphs:
      		\begin{itemize}
      			\item if Spoiler chooses a vertex, then Duplicator chooses any vertex
      			(belonging to the equivalent set if necessary);
      			\item if Spoiler chooses a set $S$ for example in $A$, then Duplicator choose 
      			a set of the same size in $B$ if $|S| \le r$ or $|V(A) \backslash S| \le r$, otherwise
      			Duplicator chooses any set of size $|S|$ in $\{r+1, \dots, |V(B)|-r-1\}$ 
			(for large enough $N$ the graph $B$ is big enough, \emph{i.e.},
      			$|V(B)|-r-1 \ge r+1$).
      		\end{itemize}
		After $r$ rounds, all the vertices and sets that have been chosen are partially isomorphic
      		in $A$ and $B$, since they form either cliques or independent sets,
		and there are too many vertices in $A$ and $B$ in order to count them
      		in only $r$ rounds.
		By \cite[Corollary~7.10]{libkin2004elements}, a property distinguishing $A$ 
      		and $B$ is not expressible in MSO-logic, this is a contradiction. We conclude that
      		if $A$ and $B$ are both cliques or both independent sets, then they are both models
      		or both counter-models of $\psi$.
      		
      		Now that each of the sets $\{K_n \mid n > N_r\}$ and $ \{I_n \mid n > N_r\}$ have been 
      		shown to be exclusively models or exclusively counter-models,
		by non-triviality of $\psi$ one of them are models, and the other are counter-models.
      	\end{proof}
      	
      	\begin{proof}[Proof of Proposition~\ref{prop:pcomplete}]
		For any cw-non-trivial $\psi$, the couple $(\psi,\xi)$ is cw-size-independent
		by Lemma~\ref{lemma:cliqueindep} since for any $n > N_r$, among the graphs $K_n$ 
		and $I_n$, there is one model and one counter-model.
		Consequently, \textbf{$\psi$-under-$\xi$-dynamics} is $\Poly$-hard by Theorem~\ref{thm:main-P-hard}. 
		Given a succinctly encoded graph $G$ of size $N$,
		knowing that it is either a clique or an independent set, and
		using properties from Lemma~\ref{lemma:cliqueindep} (where we suppose that
		ultimately all the cliques are models of $\psi$, by symmetry with the other case), we
		can decide whether $G \vDash \psi$ as follows:
		\begin{itemize}
			\item if $G$ has size lesser or equal to $N_r$, then check whether it belongs to
			  the finite list of models of $\psi$ of size less than $N_r$, and answer accordingly
			  (constant time);
			\item if $G$ has size greater than $N_r$, then
			  evaluate the circuit to check if the arrow $0 \to 0$ exists in $G$,
			  answer true if it exists and false otherwise
			  (done in polynomial time);
		\end{itemize}
		We conclude the the problem is solvable in polynomial time.
      	\end{proof}
      	
	The statement is not vacuous, as there are formulae
	which are cw-non-trivial under restriction $\xi$,
	for example
	$\psi_1 = \exists x, \exists y, x \to y$ and
	$\psi_2 = \forall x, \forall y, x \to y$.

      	
        \subsection{Extensions of the result}


	It is possible to add unary relations $R_1$ to the signature $\mathcal{S}$
	and to the structures, without affecting the constructions and results
	leading to \cref{thm:main-NP-hard,thm:main-P-hard}.
	These unary relations are encoded in a circuit with
	$\lceil\log_2(N)\rceil$ inputs and $|R_1|$ outputs.


	For any $q \geq 2$, the restriction to graphs of size $q^n$ for some $n$
	(corresponding to $q$-uniform ANs)
	is meaningful in the context of discrete dynamical systems,
	but cannot be expressed in MSO.
	Nevertheless, the study on the size of models and counter-models obtained by pumping
	(arithmetic progressions) presented in~\cite{glp25v2} still holds,
	and \cref{thm:main-NP-hard,thm:main-P-hard} remain identical under the restriction to $q$-uniform ANs
	for any fixed $q \geq 2$.

	\subsection{Is there a loop on the minimum vertex?}
	Recall that succinct graph encodings very naturally correspond to automata networks,
	vertices being AN configurations and edges being transitions.
      	One of the goals of adding a restriction~$\chi$ in our results was to enforce
	some conditions which are very natural from this AN point of view.
	For instance, one may require the AN to be deterministic,
	i.e.\ require the (directed) graph to have out-degree exactly one on all vertices,
	and thus recover results from~\cite{ggpt21}.

	In an AN, one typically considers given some initial configuration~0.
	It is natural to ask problems about this~0, e.g.\ whether it is a fixed point
	(meaning whether vertex~0 in the graph has a loop).
	There are however several different ways to interpret this initial configuration~0,
	and the MSO restriction~$\chi$ is useful to express some of them.
	First the vertex~`0' can simply mean the one encoded by the bitstring 0\dots0.
	Then, testing if this vertex has a loop is easily $\Poly$-complete,
	as it is precisely equivalent to evaluating the given circuit on the zero input.
	This problem however does not fit the \textbf{$\psi$-dynamics} setting,
	as a formula~$\psi$ cannot refer to the bitstring encoding of vertices.

	An alternative is to add a second symbol~$\le$,
	use the restriction~$\chi$ to require~$\le$ to be a linear ordering,
	and inquire about the minimum of~$\le$
	(such a $\chi$ is not union-stable, hence not in the scope of Theorem~\ref{thm:main-NP-hard}).
      	However, the way to identify configuration $0$ in terms of logic is essential here, 
      	and can lead to different complexity results as shown below.

      	\paragraph*{Minimum configuration.} If $0$ is a syntactic sugar for the 
      	global minimum configuration, given some relation $\le$ restrained by $\chi$ to be an order, 
      	then the effective formula $\psi$ would be 
        $\exists x, (\forall y, x \le y) \wedge x \to x$.
        In this case, \textbf{$\psi$-under-$\chi$-dynamics} is hard,
        because intuitively it first requires to find the minimum configuration $x$ to then check that it has a loop
        (the existence of the minimum is ensured by the promise $\chi$).
        Formally, in this case we prove the following result.

        \begin{proposition}\label{prop:minimum}
          Let $\psi=\exists x, (\forall y, x \le y) \wedge x \to x$
          and $\chi$ express that $\leq$ is a total order,
          with $R_2=\{\to,\leq\}$. Then \textbf{$\psi$-under-$\chi$-dynamics}
          is $\NP$-hard and $\coNP$-hard.
        \end{proposition}

        \begin{proof}
          Let $S$ be an instance of \textbf{SAT} on $n$ variables.
          We construct the circuit $D$ of a graph on $N=2^n$ vertices,
          such that a vertex label $x$ on $n$ bits is interpreted as a valuation
          and we denote $S(x)\in\{0,1\}$.
          Given two vertex labels $x$ and $y$, 
          circuit $D$ outputs two bits (one for $\to$ and one for $\leq$) as follows:
          \begin{itemize}
            \item if $x=y$, then output $S(x)$ for relation $\to$;
            \item if $x\neq y$, then output $0$ for relation $\to$;
            \item if $S(x)=1$ and $S(y)=0$, then output $1$ for relation $\leq$;
            \item if $S(x)=0$ and $S(y)=1$, then output $0$ for relation $\leq$;
            \item if $S(x)=S(y)$, then output the integer comparison of $x$ and $y$ for relation $\leq$.
          \end{itemize}
          Relation $\leq$ is indeed a total order, with vertex labels interpreted as integers
          but with all satisfying valuations lesser than all falsifying valuations of $S$.
          Given that only the satisfying valuations bear a loop,
          $S$ is satisfiable if and only if the minimum vertex has a loop.
          The circuit $D$ can be produced in polynomial time, hence this proves the $\NP$-hardness.
          For the $\coNP$-hardness, simply put the loops on falsifying valuations.
        \end{proof}
      	
      	\paragraph*{Unary predicate.}
	An alternative encoding of a distinguished configuration using only MSO definable constraints
	is to add a unary predicate $\mathsf{zero}$, and restrict it with $\chi$ by allowing
        exactly one vertex to verify $\mathsf{zero}$. We shall call that vertex~`0'.
        Formula $\chi$ remains up-to-isomorphism, hence deciding whether
        $\psi = 0 \to 0$ holds still requires to find which configuration $x$ verifies $\mathsf{zero}(x)$,
        to then check that is has a loop.
        Nevertheless, in this case the promise $\chi$ provides an easy check, and we have the following.

        \begin{proposition}\label{prop:zero}
          Let $\psi=\exists x,\mathsf{zero}(x)\wedge x\to x$
          and $\chi=\exists!x,\mathsf{zero}(x)$,
          with binary predicate $\to$ and unary predicate $\mathsf{zero}$.
          Then \textbf{$\psi$-under-$\chi$-dynamics} is in $\NP\cap\coNP$.
        \end{proposition}

        \begin{proof}
          For the $\NP$ upper bound: guess $x$ and then check that $\mathsf{zero}(x)$ and $x\to x$.
          For the $\coNP$ upper bound, the complementary problem is \textbf{$\neg\psi$-under-$\chi$-dynamics}
          (with formula $\neg\psi=\forall x,\neg\mathsf{zero}(x)\vee\neg x\to x$
          and the same promise that $\chi$ holds).
          A polynomial time non-deterministic algorithm is simply:
          guess $x$ and then check that $\mathsf{zero}(x)$ and $\neg x\to x$.
          It correctly solves \textbf{$\neg\psi$-under-$\chi$-dynamics} because
          there is a unique $x$ such that $\mathsf{zero}(x)$, hence for this $x$
          we must have $\neg x\to x$ for the instance to be positive.
        \end{proof}

      	\paragraph*{Explicit AN configurations.}
	The previous two encodings of the initial configuration~0 are expressed only with the MSO~constraint~$\chi$.
	Thus, they are inherently up to isomorphism, and independent of the bitstring encoding of this configuration~0.
	This may seem natural from a succinct encoding point of view, where the bitstring encoding is mainly an implementation detail,
	but from the AN perspective, this bitstring becomes much more meaningful,
	and one may for instance care about configurations where all bits are the same.

        One way to circumvent this limitation may be to add outside constraints, not expressible in logic.
	For instance, by adding a constant~0 (which~$\psi$ can refer to), and requiring it to be interpreted as the vertex with bitstring~0,
	or adding a predicate~$\le$ and requiring it to be interpreted as the lexicographic ordering on bitstrings.
	An alternative may be to consider a family of unary predicates~$\mathsf{s}_i$ for $i\in\N$ where~$\mathsf{s}_i(x)$ indicates that the~$i$th bit of~$x$ is~1.
	The issue of course is that this is an unbounded number of predicates, depending on the circuit size.
        The logics should thus allow to quantify, or fix, values for $i$.
        Then configuration $0$ is expressed as $z$ in $\forall i\in\N, \mathsf{s}^0_i(z)$.

	We leave as open question the investigation of meta-hardness results under such constraints that are not expressed in logic.
	It is likely that such result would be highly dependent on the specific choice of this constraint.

	\section{On bounded twin-width and bounded degree}
        \label{s:bounded}

        Characterizing the complexity of deciding a non-trivial yet cw-trivial 
        formula is still open. Indeed, with MSO formula, we do not always have a set
        of models that is of bounded clique-width. Regarding this matter,
	\cite[Theorem~7.2]{gglgopt25} shows that there are non-trivial formula~$\psi$ for which,
        under a reasonable complexity assumption, \textbf{$\psi$-dynamics} cannot be \NP or \coNP hard.
	(\textbf{$\psi$-dynamics} is just \textbf{$\psi$-under-$\chi$-dynamics} for $\chi=\top$,
	that is to say without restriction).

	The complexity assumption used in \cite[Theorem~7.2]{gglgopt25} uses the notion of
	\emph{robustly often polynomial} algorithms of Fortnow and Santhanam~\cite{FS2017robust}.
	\begin{definition}[Robustness]
		A set $M$ of integers is called \emph{robust} if:
		$$\forall k, \exists \ell \ge 2: \{\ell, \ell+1, ..., \ell^k\} \subset M.$$
		In particular, a robust set is infinite and contains arbitrarily long intervals of integers.
	\end{definition}
	A problem~$P$ is called \emph{robustly often polynomial} if there is a robust set~$M$
	and a polynomial time algorithm which solves~$P$ for any instance whose size is in~$M$.
        \begin{theorem}[{\cite[Theorem~7.2]{gglgopt25}}]%
        	\label{theorem:cwtrivial}
                There is a non-trivial first-order sentence~$\psi$ such that,
		if \textbf{$\psi$-dynamics} is \NP- or \coNP-hard,
		then \textbf{SAT} is robustly often polynomial.
        \end{theorem}
	That is, encoding \textbf{SAT} into this \textbf{$\psi$-dynamics} problem
	would yield a polynomial time algorithm for \textbf{SAT} for a robust set of input sizes,
        which is unlikely to be possible.
	Thus the \NP- or \coNP-hardness result of \cref{thm:main-NP-hard} can likely not
	be generalised to formulae~$\psi$ which are merely non-trivial (instead of cw-non-trivial),
	even if one only considers FO formulae.

	Then, we can question if similar \NP- or \coNP-hardness can be proved
	under some hypothesis intermediate between cw-non-triviality, and just non-triviality.
	Our next result shows that this is likely not possible:
	\cref{theorem:cwtrivial} still holds even when considering only planar bounded degree graphs.
	Remark that it implies that clique-width in \cref{thm:main-NP-hard} cannot be replaced by another graph parameter such as \emph{twin-width}, because planar graphs
        have bounded twin-width \cite{BKTW2022twinwidth,HJ2023twinwidth}.

	Let~$\Pc_4$ denote the class of all planar graphs with maximum degree at most~4.
	\begin{theorem}[\cref{thm:intro-planar} restated]\label{theorem:nontrivial}
	  There is a first-order formula~$\psi$ that is $\Pc_4$-non-trivial,
	  such that if \textbf{$\psi$-dynamics} is \NP- or \coNP-hard,
	  then \textbf{SAT} is robustly often polynomial.
	\end{theorem}

        As in the reasoning inspiring the current proof \cite{gglgopt25}, 
        the notion of spectrum is crucial (Definition~\ref{def:spectrum}).
        We use the following result of Dawar and Kopczyński.
        \begin{theorem}[{\cite[Theorem~7.1]{dawar2024spectra}}]
                \label{thm:planar-spectrum}
                Let~$S \subset \N$ be such that there is a non-deterministic Turing machine that,
		given the binary representation of~$N$,
                tests whether~$N \in S$ in time~$O(N^{1-\epsilon})$ and space~$O(\log N)$.
                Then there is a first-order sentence~$\psi_S$ over the language of colored
		graphs\footnote{In~\cite{dawar2024spectra}, the sentence~$\psi_S$ uses two kinds of directed edges.
		This can be replaced by adding a bounded number of unary predicates.}, such that:
                \begin{enumerate}
                        \item the spectrum of~$\psi_S$ is exactly~$S$; and
                        \item every model of~$\psi_S$ is a colored planar graph with maximum degree at most~4.
                \end{enumerate}
        \end{theorem}

        Call a function~$h$ \emph{time-constructible} if there is a Turing machine that,
        given~$n$ in unary, computes~$h(n)$ in time~$O(h(n))$.
        The hypothesis on~$S$ in Theorem~\ref{thm:planar-spectrum} is satisfied
        by the image of any time-constructible function
        up to an exponential allowing some margin in the time and space usage
        (compare with \cite[Lemma~7.9]{gglgopt25}).

        \begin{corollary}\label{cor:planar-spectrum-constructible}
          Let~$h$ be a time-constructible map with~$h(n) \ge n$.
          Then there is a first-order sentence~$\psi_h$ over the language of colored graphs such that:
          \begin{enumerate}
            \item the spectrum of~$\psi_h$ is $\exp(h(\N))$; and
            \item every model of~$\psi_h$ is a colored planar graph with maximum degree at most~4.
          \end{enumerate}
        \end{corollary}

        \begin{proof}
          Using Theorem~\ref{thm:planar-spectrum}, it suffices to show that given~$N$ in binary, one can test whether $N \in \exp(h(\N))$ in time~$\O(\log N)$.
          First check that~$N$ is a power of~2, and compute~$\log N$.
          Then guess~$k \le \log N$ (in unary), and check that $h(k) = \log N$.
          Since we assumed~$n \le h(n)$, there is such a~$k$ if and only if $\log N \in h(\N)$,
          and the previous steps can clearly be performed in time~$\O(\log N)$.
        \end{proof}

        We can now strengthen \cite[Lemma~7.9]{gglgopt25} as follows.

        \begin{corollary}
        	\label{cor:PrimRecMapSpec}
          There is a first-order sentence~$\psi_h$ such that:
          \begin{enumerate}
            \item $\spec(\psi_h)$ is infinite;
            \item for any increasing primitive recursive map~$R$, $R(\N) \not\subseteq \spec(\psi_h)$; and
            \item every model of~$\psi_h$ is a colored planar graph with maximum degree at most~4.
          \end{enumerate}
        \end{corollary}

        \begin{proof}
          Take~$A(n)$ to be the Ackermann function, and write $f(n) = \exp(A(n))$.
          The Ackermann function is increasing, time-constructible, and for any primitive recursive~$R$, there is~$n$ with $A(n) > R(n)$, and a fortiori $\exp(A(n)) > R(n)$.
          When~$R$ is injective, this implies $R(\N) \not\subseteq f(\N)$.
          Indeed, injectivity gives that $R(\{1,\dots,n\})$ cannot be contained in $f(\{1,\dots,n-1\})$,
          and for any~$m \ge n$, we have $f(m) \ge f(n) > R(\{1,\dots,n\})$.

          Applying Corollary~\ref{cor:planar-spectrum-constructible} to~$A$,
          we obtain a first-order formula~$\psi$ with $\spec(\psi) = \exp(A(\N))$,
          and whose models are planar with degree at most~4,
          hence proving the statement.
        \end{proof}

        The last tool we need is the following.
	Consider a reduction~$f$ from \textbf{SAT} (or \textbf{UNSAT}) to \textbf{$\psi$-dynamics}.
	For~$S$ a \textbf{SAT} (or \textbf{UNSAT}) instance, $f(S)$ is a circuit representing a graph $G_{f(S)}$.
	Recall that $|G_{f(S)}|$ denotes the number of vertices.
	\begin{definition}[{\cite[Definition~7.4]{gglgopt25}}]
	  The reduction~$f$ is \emph{meager at~$n \in \N$} if
	  for any \emph{positive} \textbf{SAT} instance~$S$ of size~$n$, $|G_{f(S)}| \le n$.
	  The \emph{meagerness} set of~$f$ is the set of integers~$n$ at which~$f$ is meager.
	\end{definition}
	Note that in \cite{gglgopt25}, meagerness is parameterized by a polynomial~$P$,
	However, one can freely choose~$P$ in their proof; the previous definition corresponds to the case~$P(n) = n$.

	The size of~$f(S)$ is polynomial in~$S$, hence~$|G_{f(S)}|$ can be up to exponential in~$S$.
	The meagerness set thus indicates sizes for which any positive instance is mapped to a `small' graph.
	This implies that \textbf{SAT} is easy for meager sizes:
	\begin{lemma}[{\cite[Lemma~7.5]{gglgopt25}}]\label{lem:meager-algo}
	  Let~$f$ be a polynomial reduction from \textbf{SAT} (or \textbf{UNSAT}) to \textbf{$\psi$-dynamics}, and $M$ the meagerness set of~$f$.
	  Then there is a polynomial algorithm that solves \textbf{SAT} for all instance sizes in~$M$.
	\end{lemma}
	The next result of~\cite{gglgopt25} relates the spectrum of~$\psi$ to the meagerness set of its reduction~$f$.
	\begin{lemma}[{\cite[Lemma~7.6]{gglgopt25}}]\label{lem:robust-prim-rec}
	  Let~$f$ be a polynomial reduction from \textbf{SAT} (or \textbf{UNSAT}) to \textbf{$\psi$-dynamics}. Then
	  \begin{enumerate}
	    \item either the meagerness set of~$f$ is robust, or
	    \item there is an increasing primitive recursive sequence~$\mu$ such that~$\psi$ has a model of size~$\mu(n)$ for all~$n$.
	  \end{enumerate}
	\end{lemma}
	The statement of the previous lemma in~\cite{gglgopt25} is somewhat more complex,
	parameterized by an integer $d \ge 1$; the previous result corresponds to the case $d=1$.

	We are now ready to prove the main result of this section.
        \begin{proof}[Proof of Theorem~\ref{theorem:nontrivial}]
          Let $\psi_h$ be given by Corollary~\ref{cor:PrimRecMapSpec}.
          Conditions~1 and~3~from Corollary~\ref{cor:PrimRecMapSpec}
          give that~$\psi_h$ has infinitely many models in~$\Pc_4$.
          On the other hand, condition~2 implies that the complement of~$\spec(\psi_h)$ is infinite
          (by applying the condition to the map $n \mapsto n+k$ for each~$k$).
          Thus for infinitely many~$n$, \emph{all} graphs of size~$n$ are models of~$\lnot\psi_h$.
          Since~$\Pc_4$ contains graphs of all possible sizes, it follows that $\lnot\psi_h$ has infinitely many models in~$\Pc_4$.
          Thus~$\psi_h$ is $\Pc_4$-non-trivial.

	  Now assume \textbf{$\psi_h$-dynamics} is \NP- or \coNP-hard.
	  Thus there is a polynomial reduction~$f$ from \textbf{SAT} or \textbf{UNSAT} to \textbf{$\psi_h$-dynamics}.
	  Consider its meagerness set~$M$.
	  By \cref{lem:meager-algo}, there is a polynomial algorithm solving \textbf{SAT} for all instance sizes in~$M$,
	  hence if~$M$ is robust, then \textbf{SAT} is robustly often polynomial.
	  On the other hand, if~$M$ is not robust, then \cref{lem:robust-prim-rec} gives an increasing primitive recursive sequence~$\mu(n)$ contained in~$\spec(\psi_h)$,
	  contradicting condition~2 from \cref{cor:PrimRecMapSpec}.
        \end{proof}

        \section*{Acknowledgments}
        
        The authors thank Pierre Ohlmann for many helpful discussions and de example of cw-size-dependent formulae.
        This work received financial support from projects
        ANR-24-CE48-7504 ALARICE and
        HORIZON-MSCA-2022-SE-01 101131549 ACANCOS.
      
\bibliographystyle{plainurl}
\bibliography{biblio}

\end{document}